\documentclass[11pt]{article}
\usepackage[top=1.16in, bottom=1.16in,left=1.16in,right=1.16in]{geometry}
\usepackage{graphicx}
\usepackage{amsmath,amssymb, endnotes, amsfonts, amsthm}
\usepackage{epstopdf}
\usepackage{setspace}
\def\sA{{\cal A}}

\def\sF{{\cal F}}

\def\bI{\mathbb{I}}
\def\R{\mathbb{R}}
\def\bF{\mathbb{F}}
\def\E{\mathbb{E}}

\def\bP{\mathbb{P}}

\def\bVN{\mathbb{V}^{(N)}}     
\def\VNinv[t]{\left(\bVN_t\right)^{-1}}
\def\sgn{{\text{sign}}}          
\def\hatlambda{\widehat\lambda}
\def\Shat{\widehat S}          
\def\lambdamax{\lambda^{\text{max}}}

\renewcommand{\labelenumi}{(\alph{enumi})}

\numberwithin{equation}{section}
\theoremstyle{plain}                
\newtheorem{theorem}{Theorem}[section]
\newtheorem{lemma}[theorem]{Lemma}
\newtheorem{proposition}[theorem]{Proposition}

\theoremstyle{definition}           
\newtheorem{definition}[theorem]{Definition}
\newtheorem{example}[theorem]{Example}

\theoremstyle{remark}               

\begin{document}

\begin{center}
{\large \bf Price impact equilibrium with transaction costs and TWAP trading} \\ 
Eunjung Noh and Kim Weston\footnote{The authors would like to thank Xiao Chen and Kasper Larsen for helpful discussions on this work. The second author acknowledges support by the National Science Foundation under Grant No. DMS\#1908255 (2019-2022). Any opinions, findings and conclusions or recommendations expressed in this material are those of the author and do not necessarily reflect the views of the National Science Foundation (NSF).}
\ \\
Rutgers University \\
Department of Mathematics
\ \\ \ \\

\today
\end{center}
\vskip .5in

\abstract{We prove the existence of an equilibrium in a model with transaction costs and price impact where two agents are incentivized to trade towards a target.  The two types of frictions -- price impact and transaction costs -- lead the agents to two distinct changes in their optimal investment approach:  price impact causes agents to continuously trade in smaller amounts, while transaction costs cause the agents to cease trading before the end of the trading period.  As the agents lose wealth because of transaction costs, the exchange makes a profit.  We prove the existence of a strictly positive optimal transaction cost from the exchange's perspective.}


\section{Introduction}
We study a financial equilibrium model with frictions stemming from both transaction costs and price impact.  Two agents are incentivized to trade towards a targeted number of shares throughout the trading period.  In equilibrium, the agents seek to maximize their expected wealth minus a penalty for deviating from their targets.  Their wealth is further reduced by transaction costs and is affected by the perceived price impact on the stock price from their trades.

Incomplete equilibrium is notoriously difficult to study.  When the incompleteness stems from frictions, this difficulty is exacerbated.  This work proposes a tractable model for a financial equilibrium with two simultaneous frictions.  We seek to answer two questions:
\begin{enumerate}\renewcommand{\labelenumi}{(\arabic{enumi})}
  \item \emph{How do transaction costs and price impact affect prices and strategies in equilibrium?}
  \item \emph{What is the optimal level of transaction costs?}
\end{enumerate}
The main contribution of this work is that the effect of transaction costs in equilibrium is distinct from that of price impact.  Price impact affects the equilibrium stock price and depresses the rate of trade, whereas transaction costs cause the agents to cease trading early in the trading period.  We also prove the existence of a strictly positive optimal level of transaction costs, where optimality is determined from the perspective of the market's exchange. The exchange collects the transaction fees as the agents trade.



The two frictions that we focus on are transaction costs and price impact.  When studied individually, transaction cost equilibrium models often resort to simplifications in order to draw conclusions.  Our model shares this approach, as it is simple enough to remain tractable while complex enough to capture differing effects of our two frictions.  Weston~\cite{W18MAFE} proves the existence of a transaction cost equilibrium in a tractable model with deterministic equilibrium annuity prices.  Continuum-of-agent models, where market clearing is averaged over infinitely many agents, are studied in Vayanos and Vila~\cite{VV99ET}, Vayanos~\cite{V98RFS}, Huang~\cite{H03JET}, and D\'avila~\cite{D20wp}.  Herdegen and Muhle-Karbe~\cite{HMK18FS} study equilibrium with transaction costs where clearing holds approximately, up to a leading order.  For proportional transaction costs, Gonon et.\,al.~\cite{GMKS19wp} study equilibrium with an ergodic objective.  Lo et.\,al.~\cite{LMW04JPE} and Buss et.\,al.~\cite{BUV14wp} study the numerics behind transaction cost equilibria without establishing the existence of an equilibrium.  In contrast, we prove the existence of an equilibrium with proportional transaction costs (and price impact) in a standard setting with two of agents, market clearing, and consumption at a terminal time.

Price impact and optimal liquidation models for a single agent take the price impact form as given and allow agents to consider how the size and timing of their trades will impact the traded asset price and, hence, future wealth.  Our equilibrium model endogenizes price impact, which is realized in the stock's drift in the form proposed in Cuoco and Cvitani{\'{c}}~\cite{CC98JEDC}.  We incorporate price impact in our model by modeling it through off-equilibrium price paths.  Off-equilibrium paths describe how prices react to trading strategies even though a given trading strategy may be suboptimal.  The optimal trading strategy determines the on-equilibrium price path.  Both agents' on-equilibrium price dynamics must coincide with the equilibrium stock price.  In this way, price impact affects equilibrium prices, even though it is only modeled through each individual agent's perceived off-equilibrium price paths.

Other models have endogenized price impact in an equilibrium setting using various approaches.  Kyle~\cite{K85E} uses a game theoretic framework, where the market maker attempts to filter out private information from the aggregate trades of noise traders and an informed trader.  Vayanos~\cite{V99RES} and Sannikov and Skrzypacz~\cite{SS16wp} use private trading target information to endogenize price impact.  Choi et.\,al.~\cite{CLS19wp} endogenize price impact by allowing off-equilibrium price paths to vary in a Cuoco and Cvitani{\'{c}}~\cite{CC98JEDC}-sense while also allowing their agents' trading targets to be private.

Our modeling set-up is most similar to Choi et.\,al~\cite{CLS19wp}, who study targeted trading in Nash equilibria with price impact.  We introduce transaction costs into a reduced form of their setting 
in order to compare the effects of both frictions (price impact and transaction costs) in equilibrium. 
Our model shares similarities with others in terms of single-agent models and equilibrium settings.  Linear-quadratic models with trading targets are studied in several works, including Bank et.\,al.~\cite{BSV17MAFE}, Sannikov and Skrzypacz~\cite{SS16wp}, and Vo\ss~\cite{V19wp}. Gonon et.\,al.~\cite{GMKS19wp} use linear-quadratic controls and allow for proportional transaction costs, rather than a quadratic approximation to transaction costs, as in Brunnermeier and Pedersen~\cite{BP05JF}, G\^arleanu and Pedersen~\cite{GP16JET}, and Bouchard et.\,al.~\cite{BFHMK18FS}.

Every transaction fee paid by the agents is income for the market's exchange.  Higher transaction costs generate more income for the exchange for every share traded, but lower transaction costs induce the agents to trade a higher volume of shares.  Consequently, we prove that there exists a strictly positive level of transaction costs that maximizes the exchange's expected profit.  Optimal transaction costs have been of interest starting most notably with the introduction of the Tobin tax in Tobin~\cite{T78EEJ}.  Previous equilibrium approaches consider optimality from a welfare perspective.  The continuum of agents in D\'avila~\cite{D20wp} differ in their beliefs about the dividend's distribution.  The agents' belief difference versus the central planner's choice of distribution when calculating welfare leads to a strictly positive optimal transaction tax.  In Weston~\cite{W18MAFE}, the welfare decreases as transaction costs increase, leading to zero as the welfare optimizing transaction cost parameter.  In our model, the agents are identical in their beliefs and differ only in their trading targets.  Nonetheless, we prove the existence of a strictly positive optimal transaction cost from the exchange's perspective.

The paper is organized as follows.  Section~\ref{section:setting} describes our model inputs.  Section~\ref{section:eq} presents our main result, Theorem~\ref{thm:main}, which establishes the existence of a price impact equilibrium with transaction costs.  The choice of an optimal transaction cost from the perspective of the exchange is presented in Section~\ref{section:optimal-lambda}.  The proofs are contained in Section~\ref{section:proofs}.



\section{The model}\label{section:setting}

Let $T>0$ be a fixed time horizon, which we think of as one trading day in length.  We work in a continuous-time setting and let $B=\{B_t\}_{t\in[0,T]}$ be a Brownian motion on a probability space $(\Omega,\sF,\bP)$.  The market consists of two traded securities:  a bank account and stock.  The bank account is a financial asset in zero-net supply with a constant zero interest rate.  The stock is in constant positive net supply with the supply denoted by $n>0$.  It pays a dividend of $D$ at time $T$.  The random variable $D$ is measurable with respect to $\sigma\left(B_u\!:\, 0\leq u\leq T\right)$ and $\E[D^2]<\infty$.  We let $\sigma$ be the progressively measurable process such that $\E\int_0^T\sigma_{u}^2du<\infty$ and
$$
  D= \E[D] + \int_0^T \sigma_{u}dB_u
$$
guaranteed by the martingale representation theorem.  The equilibrium stock price will be determined endogenously in equilibrium and is denoted $\Shat = \{\Shat_t\}_{t\in[0,T]}$.  The price is constrained at time $T$ so that $\Shat_T = D$.  We assume that all goods are denominated in a single consumption good.

Two investors, $i=1,2$, trade in the market.   They each seek to maximize expected wealth yet are subjected to inventory penalties throughout the trading period.  Their wealth is further penalized by transaction costs, which are proportional to the rate of trade at the rate $\lambda>0$.  Their wealth is indirectly penalized by perceived price impact from trades.  For every share purchased, the agents perceive that the stock's drift decreases linearly.

Each agent $i$ has a target number of shares $a_i$ she wishes to acquire (or sell off) throughout the trading period.  The random variables $a_1$ and $a_2$ are assumed to satisfy $\E[a_i^2]<\infty$ and be independent of the Brownian motion $B$.  Agent $i$'s filtration $\bF_i = \{\sF_{i,t}\}_{t\in[0,T]}$ is given by
$$
  \sF_{i,t} := \sigma\left(a_i, \Shat_u:\, u\in[0,t]\right), \ \ \ \ t\in[0,T].
$$

  A trading strategy $\theta = \{\theta_t\}_{t\in[0,T]}$ denotes the number of shares held in stock.  For $i\in\{1,2\}$, we say that $\theta$ is \textit{admissible for agent $i$} if it is adapted to $\bF_i$, c\`adl\`ag, of finite variation on $[0,T]$ $\bP$-a.s., and satisfies $\E\int_0^T \left(\sigma_{t}\theta_t\right)^2dt <\infty$.  We write $\sA_i$ to denote the collection of agent $i$'s admissible strategies.  Agent $i$ is endowed at the beginning of the trading period with $n/2$ shares of stock. We normalize the shares in the bank account to zero.  We allow for $\theta_0$ to differ from $n/2$, as the agents may choose to trade a lump sum immediately.  In the absence of transaction costs or the penalty term given in~\eqref{def:L} below, the agents' allocations would be Pareto optimal.  However, the presence of frictions and penalties will motivate the agents to deviate from their initial positions.
  
  For $i=1,2$, since $\theta\in\sA_i$ is of finite variation, we decompose $\theta$ into 
  \begin{equation}\label{eqn:decomp}
    \theta_t = \frac{n}{2} + \theta^\uparrow_t - \theta^\downarrow_t, \ \ \ \ t\in[0,T],
  \end{equation}
  where $\theta^\uparrow, \theta^\downarrow$ are adapted to $\bF_i$, c\`adl\`ag, nondecreasing, 
  and 
\begin{equation}\label{eqn:trade-condition}
  \{t\in[0,T]:d\theta^\uparrow_t>0\}\cap\{t\in[0,T]:d\theta^\downarrow_t>0\}=\emptyset.
\end{equation}
  A change in trading position is possible at time $0$, and we allow for $\theta^\uparrow_0>0$ or $\theta^\downarrow_0>0$ as long as~\eqref{eqn:trade-condition} holds.

At the close of the trading period, agents consume their acquired dividends.  They are  subjected through their optimization problems to inventory penalties throughout the trading period.  For $i\in\{1,2\}$ and a given $\theta\in\sA_i$, the penalty term, or loss term, for agent $i$ is measured by
\begin{equation}\label{def:L}
  L_{i,T}^\theta:= \frac12\int_0^T \kappa(t)\left(\gamma(t)\left(a_i-\frac{n}{2}\right) - \left(\theta_t-\frac{n}{2}\right)\right)^2dt.
\end{equation}
The function $\kappa:(0,T)\rightarrow (0,\infty)$ describes the intensity of the penalty, while $\gamma:[0,T]\rightarrow[0,1]$ describes the desired intraday trading target trajectory.  Both agents share the same deterministic functions $\kappa$ and $\gamma$.  We assume that $\gamma$ is c\`adl\`ag, nonnegative, bounded in absolute value by one, and nondecreasing.  Our main example is time-weighted average price (TWAP), where the intraday trajectory function is $\gamma^{\text{TWAP}}(t):= t/T$.  We assume that $\kappa$ is measurable and $\int_0^T\kappa(t)dt$ is finite.

The agents perceive a price impact as the result of their trades.  For $i=1,2$ and a trading strategy $\theta\in\sA_i$, we model this impact via the off-equilibrium stock price's drift by
\begin{equation}\label{eqn:price-impact}
  dS_{i,t}^\theta = \kappa(t)\left(c_0(t;a_1+a_2) -c_1 \theta_t 
    + \gamma(t)c_2\left(a_i-\frac{n}{2}\right)\right)dt + dM_t,
    \ \ \ S_0\in\R.
\end{equation}
The function $c_0$ and constants $c_1$, $c_2$ will be determined in equilibrium and are the same for both agents.  The constant $c_1$ will be a free parameter that determines the level of price impact in the market.  The $c_1=0$ case corresponds to an equilibrium without price impact, where the agents are price-takers.  Following the work of Choi et.\,al.~\cite{CLS19wp} without transaction costs, we work with off-equilibrium stock prices whose martingale term $M$ and initial value $S_0$ are independent of $\theta$ and $i$.

Price impact in~\eqref{eqn:price-impact} is realized through the drift of the stock as in Cuoco and Cvitani{\'{c}}~\cite{CC98JEDC}.  Larger values for $c_1$ result in more price impact because the more an agent buys, the more she drives her perceived future prices down.  While traditional price impact models, such as Almgren and Chriss~\cite{AC01JR}, affect the stock price directly, our variation of price impact affects the future returns of the stock by depressing them when a trader seeks a larger market share.

For agent $i\in\{1,2\}$ and a trading strategy $\theta\in\sA_i$, agent $i$'s perceived wealth process is given by
\begin{equation}\label{def:wealth}
  X^\theta_{i,t} = \frac{n}{2}S_{0} + \int_{0}^t \theta_u dS^\theta_{iu} - \lambda\left(
  \theta^\uparrow_t + \theta^\downarrow_t\right),
  \ \ \ \ t\in[0,T].
\end{equation}
We recall that the decomposition of $\theta$ in \eqref{eqn:decomp} allows for $\theta^\uparrow_0$ and $\theta^\downarrow_0$ to differ from zero.  Both frictions -- price impact and transaction costs -- are at play in the perceived wealth dynamics.  Agent $i$'s objective is
$$
  \E\left[X^\theta_{i,T} - L^\theta_{i,T}\,\vert\,\sF_{i,0}\right] \longrightarrow \ \text{max}
$$
over $\theta\in\sA_i$, where $L^\theta_{i,T}$ is defined in~\eqref{def:L} and $X^\theta_{i,T}$ in~\eqref{def:wealth}.

\section{Equilibrium}\label{section:eq}
In an equilibrium, the stock price is determined so that markets clear when both agents  invest optimally.  The equilibrium stock price must agree with both agents' perceived prices when the optimal strategies are applied.
\begin{definition} \label{def:eq}
Let $\lambda>0$ be a given transaction cost level.  Trading strategies $\theta_1\in\sA_1$ and $\theta_2\in\sA_2$, a price process $\Shat = \{\Shat_t\}_{t\in[0,T]}$, and price impact coefficients $c_0, c_1, c_2$ form a \textit{price impact equilibrium} if
\begin{enumerate}
  \item \textit{Strategies are optimal:} For $i=1,2$, we have that
  \begin{equation}\label{eqn:optimality}
    \E\left[X^{\theta_i}_{i,T} - L^{\theta_i}_{i,T}\,\vert\,\sF_{i,0}\right] 
    = \sup_{\theta\in\sA_i}\E\left[X^\theta_{i,T} - L^\theta_{i,T}\,\vert\,\sF_{i,0}\right],
  \end{equation}
  where $L^\theta_{i,T}$ is defined in~\eqref{def:L}, $X^\theta_{i,T}$ in~\eqref{def:wealth}, and the off-equilibrium price impact stock dynamics are given in~\eqref{eqn:price-impact} with coefficients $c_0$, $c_1$, and $c_2$.
  \item \textit{Markets clear:} We have $\theta_{1,t} + \theta_{2,t} = n$ for all $t\in[0,T]$.
  \item \textit{Prices are consistent:} The equilibrium stock price process $\Shat$ is an It\^o process with $\Shat_T=D$ and for $t\in\{u\in[0,T]: d\theta^\uparrow_{1,u} + d\theta^\downarrow_{1,u} >0\}$, we have that
  $$
    S^{\theta_1}_{1,t} = S^{\theta_2}_{2,t} = \Shat_t.
  $$
  The price impact stock dynamics of $S^{\theta_1}_{1}$ and $S^{\theta_2}_{2}$ are given in~\eqref{eqn:price-impact} with coefficients $c_0$, $c_1$, and $c_2$.
\end{enumerate}
\end{definition}
Even though off-equilibrium, the agents perceive a price impact from their trades, the on-equilibrium stock price must agree with the agents' perceived prices when their optimal strategies are applied.  Item (c) of Definition~\ref{def:eq} requires the perceived prices to agree with the realized equilibrium price \textit{when trade occurs}.  (Since there are only two agents in the model, trade occurs if and only if agent 1 trades.) 
Therefore, in equilibrium, the two agents may have different perceived prices at times when they do not choose to trade.  This requirement on perceived prices in equilibrium is similar to employing shadow prices in equilibrium since an equilibrium stock price can only be uniquely identified when trade occurs; see Weston~\cite{W18MAFE}.

Market clearing in Definition~\ref{def:eq}(b) requires clearing of the stock market, however Walras' Law holds in our model in that the other markets (money market and real goods) clear as well.  
For $i=1,2$ and a given strategy $\theta\in\sA_i$ and equilibrium stock price $\Shat$, we define the realized wealth in equilibrium through its self-financing condition (see~\eqref{def:wealth}) by
$$
  \widehat{X}^{\theta}_t := \frac{n}{2}\Shat_0 + \int_0^t \theta_u d\Shat_u - \lambda\left(\theta_t^\uparrow+\theta_t^\downarrow\right), 
  \ \ \ \ t\in[0,T].
$$
We define the corresponding holdings in the money market account by
\begin{equation}\label{def:mm}
  \theta^{(0)}_t := \widehat X^\theta_t -\theta_t\Shat_t, \ \ \ \ t\in[0,T].
\end{equation}
We recall that the interest rate has been normalized to zero, since consumption only occurs at one point in time.

Lemma~\ref{lemma:walras} presents a version of Walras' Law applied to a price impact equilibrium with transaction costs.  Its proof is presented in Section~\ref{section:proofs}. Lemma~\ref{lemma:walras} shows that the money market account provides the mechanism by which transaction costs exit the economy.
\begin{lemma}\label{lemma:walras}
  For a given transaction cost $\lambda>0$, let a price impact equilibrium satisfying Definition~\ref{def:eq} be given with optimal stock holdings $\theta_1$, $\theta_2$ and equilibrium stock price $\Shat$.  For $i=1,2$, we let $\theta^{(0)}_{i,t}$ correspond to the equilibrium money market holdings at $t$ of agent $i$ with stock market strategy $\theta_i$, as in~\eqref{def:mm}.  Then the money market clears,
  \begin{equation*}\label{eqn:mma-clearing}
    \theta^{(0)}_{1,t} + \theta^{(0)}_{2,t}
    = - \lambda\left(\theta_{1,t}^\uparrow+\theta_{1,t}^\downarrow+\theta_{2,t}^\uparrow+\theta_{2,t}^\downarrow\right),
    \ \ \ \ t\in[0,T],
  \end{equation*}
  and the consumption market clears at $T$,
  \begin{equation*}\label{eqn:cons-clearing}
    \theta_{1,T}\Shat_T +\theta_{2,T}\Shat_T 
    = nD.
  \end{equation*}
\end{lemma}

To begin constructing our equilibrium, for each $i\in\{1,2\}$, we let 
$$
  a_\Sigma:=a_1+a_2 \ \ \text{ and } \ \ A_i:= a_i-\frac12 a_\Sigma.
$$
The random variables $A_i$ describe the deviation of the trading targets $a_i$ from the aggregate target $a_\Sigma$.  We note that $A_1+A_2 = 0$.

The presence of transaction costs causes the agents to stop trading before the end of the trading period.  To this end, we define the last trading time $\tau$ by
\begin{equation}\label{def:tau}
  \tau:= \inf\left\{t\in[0,T]:\, \lvert A_1\rvert\frac{1+2c_1}{1+c_1}\int_t^T \kappa(u)\Big(\gamma(u)-\gamma(t)\Big)du\leq\lambda\right\}.
\end{equation}
The time $\tau$ is a random variable valued in $[0,T)$.  We also define the random variable $\chi$ by
\begin{equation}\label{def:chi}
  \chi:= \frac{|A_1|(1+2c_1)}{1+c_1} \int_0^T \kappa(u)\gamma(u)du.
\end{equation}
The magnitude of $\chi$ will determine if trade occurs in the model or if the agents are deterred from trading by prohibitively high transaction costs.

The following theorem is our main result.  The proof of Theorem~\ref{thm:main} can be found in Section~\ref{section:proofs}.
\begin{theorem} \label{thm:main}
Let $\lambda>0$ and $c_1>-\frac12$ be given.  Suppose that $\kappa:(0,T)\rightarrow(0,\infty)$ is measurable with $\int_0^T\kappa(u)du<\infty$ and that $\gamma:[0,T]\rightarrow[0,1]$ is c\`adl\`ag, nonnegative, bounded by one, and nondecreasing.  There exists a price impact equilibrium where the price impact stock dynamics in~\eqref{eqn:price-impact} have coefficients $c_0$ and $c_2$ given in terms of $c_1$ by
\begin{equation}\label{def:cs}
  c_0(t) := c_1 n -\frac{1+2c_1}{2(1+c_1)}\gamma(t)(a_\Sigma-n)
  \ \ \text{ and } \ \ 
  c_2 := \frac{c_1}{1+c_1}.
\end{equation}
For $i=1,2$, the equilibrium holdings $\theta_i\in\sA_i$ are given for $t\in[0,T]$ by
\begin{equation}\label{def:theta}
  \theta_{i,t}:= 
  \begin{cases}
    \frac{n}{2}+\frac{A_i}{1+c_1}\,\gamma(t), & t<\tau,\\
    \frac{n}{2}+\frac{1}{\int_\tau^T\kappa(u)du}\left\{\int_\tau^T\kappa(u)\gamma(u)\frac{A_i}{1+c_1}du - \frac{\lambda\sgn(A_i)}{1+2c_1}\right\}, & t\geq\tau \text{ and } \chi>\lambda,\\
    \frac{n}{2}, & t\geq\tau \text{ and } \chi\leq\lambda,
  \end{cases}
\end{equation}
where $\tau$ is defined in~\eqref{def:tau}, $\chi$ is defined in~\eqref{def:chi}, and the $\sgn$ function convention is $\sgn(0)=0$ so that
$$
\sgn(x):=
\begin{cases}
  -1, & x<0,\\
  0, & x=0,\\
  1, & x>0.
\end{cases}
$$ 
We let $\Shat=\{\Shat_t\}_{t\in[0,T]}$ be defined by
\begin{equation}\label{def:S-hat}
  \Shat_t := \E[D] + \int_0^t\sigma_{u}dB_u + \frac12\int_t^T \kappa(u)\Big(\gamma(u)(a_\Sigma-n)-c_1 n\Big)du.
\end{equation}
Then, $\Shat$, $\theta_1$, $\theta_2$, $c_0$, $c_1$, and $c_2$ form a price impact equilibrium.  Moreover, the agents' filtrations agree in that $\sF_{1,t} = \sF_{2,t}$ for all $t\in[0,T]$.
\end{theorem}

For a given transaction cost parameter $\lambda>0$, Theorem~\ref{thm:main} shows that equilibrium is not unique.  Indeed, there exists a distinct equilibrium for every choice of price impact coefficient $c_1>-\frac12$.


\subsection{Effects of frictions in equilibrium}
Both transaction costs and price impact affect equilibrium, and each friction has its own distinct modeling characteristics and equilibrium effects.  Both frictions penalize the agents through their wealth reduction.  Transaction costs do so by directly subtracting transaction fees from wealth, while price impact does so indirectly by depressing the stock's drift with each increase in the number of shares held.

From a modeling perspective, agents are held accountable for transaction costs in equilibrium through market clearing, and their wealth decreases as a result. Price impact frictions appear only as {\it perceived} changes in off-equilibrium asset prices and wealth in the individual optimization problems.  Price impact is not incorporated explicitly into the market clearing condition, but the perceived prices must align with the realized prices in equilibrium by Definition~\ref{def:eq}(c).

Equilibrium effects of the two frictions are similar in that more frictions lead to less trade.  However, each friction has its own mechanism by which it impacts equilibrium outcomes.  First, we consider the effects of price impact, and we suppose that the level of transaction costs is fixed.  Higher (and positive) levels of the price impact coefficient $c_1$ cause the agents to trade less, while also causing a linear decrease in the equilibrium stock price.  The decrease in the perceived stock prices transfers over to the realized equilibrium stock price $\Shat$ in \eqref{def:S-hat}.  The case when $c_1=0$ corresponds to a price-taking equilibrium, in which the agents' perceived stock dynamics are not impacted by trade.

The equilibrium effects from transaction costs are perhaps more subtle.  When trade occurs in equilibrium, the trading strategies are unaffected by transaction costs.  Transaction costs determine how long agents are willing to trade by affecting $\tau$ in~\eqref{def:tau}.  Larger values of $\lambda$ produce smaller values of $\tau$, meaning that agents are not willing to continue trading if the penalty incurred by transactions costs is sufficiently large.

Moreover, the on-equilibrium stock price $\Shat$ in~\eqref{def:S-hat} is unaffected by transaction costs.  The apparent lack of an effect for $\Shat$ occurs because the on-equilibrium stock price can only be uniquely determined when trade occurs.  When trade does not occur, as is the case at the end of the trading period under transaction costs, $\Shat$ is consistent with equilibrium in that the agents will still agree not to trade using the price $\Shat$.  See D\'avila~\cite{D20wp} and Weston~\cite{W18MAFE} for a similar phenomenon in equilibrium models with transaction costs.

\section{Optimal transaction cost}\label{section:optimal-lambda} 
The agents pay transaction costs at a rate $\lambda>0$ proportional to the number of shares traded.  As the agents lose the wealth paid out in transaction costs, an exchange collects the fees.  Higher values of $\lambda$ mean that the exchange will receive more fees with every share traded.  However, higher values of $\lambda$ cause the agents to stop trading sooner, resulting in fewer shares traded.  Given a distributional estimate on the trading targets (i.e., priors), the exchange can find a strictly positive optimal transaction cost proportion that maximizes its expected profit.

In this section, we make the dependence on $\lambda$ explicit in the notation by denoting $\theta_i=\theta_i^{(\lambda)}$ and $\tau=\tau(\lambda)$.  We fix a price impact coefficient $c_1>-\frac12$. By Theorem~\ref{thm:main}, an equilibrium exists, and the total profit received by the exchange in that equilibrium is given by
$$
  \text{Profit}(\lambda)
    :=\lambda\left(\left(\theta_{1,T}^{(\lambda)}\right)^\uparrow
      + \left(\theta_{1,T}^{(\lambda)}\right)^\downarrow
      + \left(\theta_{2,T}^{(\lambda)}\right)^\uparrow
      + \left(\theta_{2,T}^{(\lambda)}\right)^\downarrow\right).
$$
By market clearing 
and the monotonicity of the optimal trading strategies, we see that
$$
  \text{Profit}(\lambda) = 2\lambda\left\lvert\theta_{1,T}^{(\lambda)} - \frac{n}{2}\right\rvert.
$$
The exchange will not have advanced knowledge of the agents' trading targets, and instead it must estimate the targets by using its ex ante priors.

Proposition~\ref{prop:optimal-lambda} asserts that a strictly positive optimal transaction cost level exists for the exchange.  The proof of Proposition~\ref{prop:optimal-lambda} is presented in Section~\ref{section:proofs}.
\begin{proposition}\label{prop:optimal-lambda}
Let $c_1>-\frac12$ be given.  Suppose that $\kappa:(0,T)\rightarrow(0,\infty)$ is measurable with $\int_0^T\kappa(u)du<\infty$ and that $\gamma:[0,T]\rightarrow[0,1]$ is c\`adl\`ag, nonnegative, bounded by one, nondecreasing, and there exists $t\in[0,T)$ so that $\gamma(t)>0$.   Also, suppose that $0<\E\left[A_1^2\right]<\infty$.  Then, there exists $\hatlambda>0$ so that
$$
  \hatlambda \in \text{Argmax} \left\{\lambda>0 :\, \E\left[\text{Profit}(\lambda)\right]\right\}.
$$
\end{proposition}

\begin{example} To illustrate the exchange's choice of optimal transaction cost level, we consider an example with TWAP traders.  We take $T:=1$, $n:=100$, $\kappa(t):=1$, $\gamma(t):=\gamma^{\text{TWAP}}(t) = t$, and we allow the price impact coefficient $c_1>-\frac12$ to be arbitrary.  For a given transaction cost level $\lambda> 0$, we have that
$$
  \text{Profit}(\lambda) = 2 \max\left(0, \frac{\lambda A_1}{1+c_1}-\lambda\sqrt{\frac{2\lambda A_1}{(1+c_1)(1+2c_1)}}\right).
$$
Each agent begins trading with $50$ shares.  The exchange estimates that agent 1 will seek to obtain a targeted number of shares that is uniformly distributed on $(50,55)$, while it believes that agent 2 will target exactly $50$ shares.  The exchange's expected profit can be computed by
$$
  \E\left[\text{Profit}(\lambda)\right] = \frac{8}{5} \int_0^{\sqrt{2.5}}\frac{\lambda y^2}{1+c_1}\max\left(0,y-\sqrt{\frac{2\lambda(1+c_1)}{1+2c_1}}\right)dy.
$$

Figure~\ref{fig:profit} plots the exchange's expected profit as a function of $\lambda$.  The four plots vary in the degree of the agents' perceived price impact, which is measured by the parameter $c_1$.  The case $c_1=0$ corresponds to a price-taking equilibrium with no price impact.  When $c_1>0$, the agents perceive a non-zero level of price impact.

\begin{figure}
  \center\includegraphics[width=\textwidth]{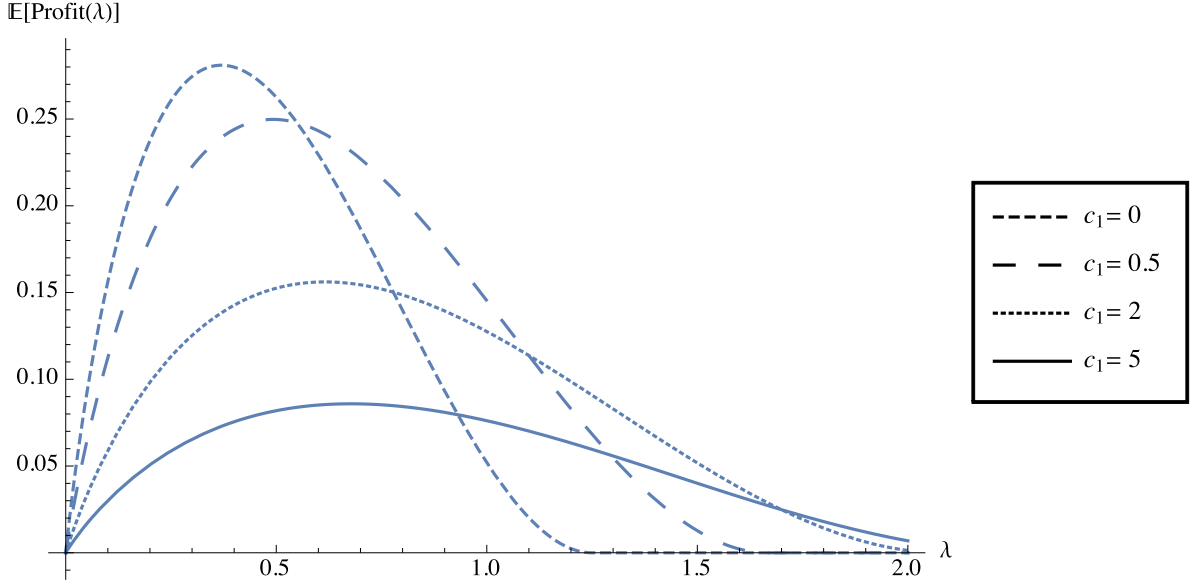}
  \caption{The exchange's expected profit in equilibrium is plotted as a function of the transaction cost $\lambda>0$.  Each plot corresponds to a different level of price impact, which is measured by varying the perceived price impact parameter $c_1>-\frac12$.}
  \label{fig:profit}
\end{figure}
\end{example}

In this example, the exchange's choice of optimal transaction cost increases with increasing price impact parameter $c_1$.  However, the exchange's optimal expected profit decreases as the perceived price impact increases.


\section{Proofs}\label{section:proofs}
We begin with a proof of Lemma~\ref{lemma:walras}, which provides a version of Walras' Law in our setting.
\begin{proof}[Proof of Lemma~\ref{lemma:walras}]
Let $\lambda>0$ be given, and let $\theta_1$, $\theta_2$, and $\Shat$ be parameters in a price impact equilibrium satisfying Definition~\ref{def:eq}.  Conditions (b) and (c) in Definition~\ref{def:eq} guarantee that the consumption market clears at $T$ by
$$
  \theta_{1,T}\Shat_T +\theta_{2,T}\Shat_T 
    = \left(\theta_{1,T}+\theta_{2,T}\right) D
    = nD.
$$
Moreover, the money market clears for all times $t\in[0,T]$ by
\begin{align*}
  \theta^{(0)}_{1,t} &+ \theta^{(0)}_{2,t}\\
  &= 2\cdot\frac{n}{2}\Shat_0 
    + \int_0^t \left(\theta_{1,u}+\theta_{2,u}\right) d\Shat_u 
    - \lambda\left(\theta_{1,t}^\uparrow+\theta_{1,t}^\downarrow+\theta_{2,t}^\uparrow+\theta_{2,t}^\downarrow\right) - \left(\theta_{1,t}+\theta_{2,t}\right)\Shat_t\\
  &= n\Shat_0 + \int_0^tn\, d\Shat_u - n\Shat_t
    - \lambda\left(\theta_{1,t}^\uparrow+\theta_{1,t}^\downarrow+\theta_{2,t}^\uparrow+\theta_{2,t}^\downarrow\right)\\
  &= - \lambda\left(\theta_{1,t}^\uparrow+\theta_{1,t}^\downarrow+\theta_{2,t}^\uparrow+\theta_{2,t}^\downarrow\right).
\end{align*}
\end{proof}

Next, we prove our main result, Theorem~\ref{thm:main}.
\begin{proof}[Proof of Theorem~\ref{thm:main}]
  We break the proof up into steps. 
  
  \noindent{\it Step 1: Information sets.} We first verify that $\sF_{1,t} = \sF_{2,t} = \sigma(a_1, a_2, \Shat_u\!:\,0\leq u\leq t)$ for all $t\in[0,T]$.  Since $\Shat_0$ given in~\eqref{def:S-hat} reveals the value of $a_\Sigma = a_1+a_2$, then for any $i\in\{1,2\}$, 
  $$
   \sF_{i,t} = \sigma\left(a_i,\Shat_u\!:\, 0\leq u\leq t\right)
     = \sigma\left(a_1, a_2, \Shat_u\!:\, 0\leq u\leq t\right),
     \ \ \ \ t\in[0,T],
  $$
  as desired.  For the remainder of the proof, we let $\bF= \{\sF_t\}_{t\in[0,T]}$ denote the shared filtration for the agents so that $\sF_{1,t}=\sF_{2,t}=\sF_{t}$ for all $t\in[0,T]$.
  
  \noindent{\it Step 2: Admissibility.} We next show for $i=1,2$, that $\theta_i$ defined in \eqref{def:theta} is admissible; that is, $\theta_i\in\sA_i$.  The function $\gamma$ is c\`adl\`ag and nondecreasing, while $\frac{A_i}{1+c_1}$ and $\tau$ are $\sF_0$-measurable.  To verify integrability, we use the square-integrability of $A_i$ and $\sigma$ as well as their independence to ensure that for some constant $C>0$, we have
  \begin{align*}
    \E\int_0^T \left(\sigma_{u}\theta_{i,u}\right)^2du
    &\leq C\, \E\int_0^T \sigma_{u}^2 \left(1+A_i^2\right) du\\
    &\leq C\, \E\left[1+A_i^2\right]\, \E\int_0^T \sigma_{u}^2du\\
    &<\infty.
  \end{align*}
  Thus, $\theta_i\in\sA_i$.
  
  \noindent{\it Step 3: Optimality.} 
  Next, we check that $\theta_i\in\sA_i$ is optimal in \eqref{eqn:optimality}. 
  In order to help us with computations, we define $\tilde\gamma$ for $t\in[0,T]$ by
\begin{equation}\label{def:gamma-tilde}
  \tilde\gamma(t):=
  \begin{cases}
    \gamma(t), & t<\tau,\\
    \frac{\int_\tau^T \kappa(u)\gamma(u)du-\frac{\lambda(1+c_1)}{|A_1(1+2c_1)|}}{\int_\tau^T\kappa(u)du}, & t\geq\tau \text{ and } \chi>\lambda,\\
    0, & t\geq\tau \text{ and } \chi\leq\lambda.
  \end{cases}
\end{equation}
The definition of $\tilde\gamma$ in comparison with $\theta_{i}$ in~\eqref{def:theta} shows that $\theta_{i,t} = \frac{n}{2} + \frac{A_i}{1+c_1}\,\tilde\gamma(t)$.  Since $A_1=0$ implies that $\chi=0$ and since $\tau<T$, we have that $\tilde\gamma$ is well-defined in the $t\geq\tau$ and $\chi>\lambda$ case.  

We define $Y_{i} = \{Y_{i,t}\}_{t\in[0,T]}$ by
$$
  Y_{i,t}:= \E\left[\int_t^T \kappa(u) \left(c_0(u) + \frac{n}{2} + \gamma(u)(1+c_2)\left(a_i-\frac{n}{2}\right) - (1+2c_1)\theta_{i,u}\right)du\,\vert\,\sF_0\right].
$$
With the form of $c_0$ and $c_2$ in \eqref{def:cs} and $\theta_i$ given in~\eqref{def:theta}, we see that
\begin{align*}
  Y_{i,t} &= A_i\frac{1+2c_1}{1+c_1}\int_t^T \kappa(u)\left(\gamma(u)-\tilde\gamma(u)\right)du.
\end{align*}
By the definition of $\tau$ in \eqref{def:tau} and $\chi$ in~\eqref{def:chi}, we have that $Y_{i,t} = \lambda\,\sgn A_i$ for $t\leq\tau$ and $\chi\geq\lambda$, while $Y_{i,t}\in~\sgn A_i[0,\lambda]$ for all $t\in[0,T]$ and all values of $\chi$.  
Since $\theta_i$ is constant after $\tau$, we see that
\begin{equation}\label{eqn:FOC-calc}
  \int_{0-}^T Y_{i,u} d\theta_{i,u} 
  = \lambda \left(\theta_{i,T}^\uparrow+\theta_{i,T}^\downarrow\right).
\end{equation}
Then, for any $\theta\in\sA_i$, 
\begin{align}
  \frac{A_i(1+2c_1)}{1+c_1}&\int_0^T\kappa(u)\left(\gamma(u)-\tilde\gamma(u)\right)\left(\theta_u-\theta_{i,u}\right)du - \lambda\left(\theta_T^\uparrow+\theta_T^\downarrow\right) \nonumber\\
  &= \int_{0-}^T Y_{i,u}d\left(\theta-\theta_i\right)_u 
     - \lambda\left(\theta_T^\uparrow+\theta_T^\downarrow\right) 
     \ \ \ \text{by integration by parts}\nonumber\\
  &= \int_{0-}^T Y_{i,u}d\theta_u 
     -\lambda\left(\theta_{i,T}^\uparrow+\theta_{i,T}^\downarrow\right) 
     - \lambda\left(\theta_T^\uparrow+\theta_T^\downarrow\right) 
     \ \ \ \text{by \eqref{eqn:FOC-calc}}\nonumber\\
  &= \int_{0-}^T\underbrace{\left(Y_{i,u}-\lambda\right)}_{\leq 0}d\theta^\uparrow_u 
     - \int_{0-}^T\underbrace{\left(\lambda+Y_{i,u}\right)}_{\geq 0}d\theta^\downarrow_u 
     -\lambda\left(\theta_{i,T}^\uparrow+\theta_{i,T}^\downarrow\right) \nonumber\\
  &\leq -\lambda\left(\theta_{i,T}^\uparrow+\theta_{i,T}^\downarrow\right).
  \label{eqn:FOC-calc-2}
\end{align}

  For an arbitrary $\theta\in\sA_i$, we define $V_i(\theta)$ by
  \begin{equation}\label{eqn:V}
    V_i(\theta):= \E\left[X^\theta_{i,T} - L^\theta_{i,T}\,\vert\,\sF_{0}\right],
  \end{equation}
  where $L^\theta_{i,T}$ is defined in~\eqref{def:L}, $X^\theta_{i,T}$ in~\eqref{def:wealth}, and the perceived off-equilibrium price dynamics in~\eqref{eqn:price-impact} have initial stock price and martingale parts given by the proposed equilibrium stock price $\Shat$ in~\eqref{def:S-hat}.
  
  By the definition of $\Shat$ in \eqref{def:S-hat}, the martingale part of $\Shat$, $S_1^\theta$, and $S_2^\theta$ is given by the martingale $\int_0^\cdot \sigma dB$.  Since $\theta\in\sA_i$ is adapted and $\E\int_0^T\left(\sigma_{u}\theta_u\right)^2du<\infty$, we have that $\left\{\int_0^t \sigma_{u}\theta_udB_u\right\}_{t\in[0,T]}$ is a square-integrable martingale with respect to $\bF$. We use this property in the calculation of $V_i(\theta)$ below.
  
  For $c_0$ and $c_2$ given in \eqref{def:cs} and the initial stock price given by $\Shat_0$ in~\eqref{def:S-hat}, we have
  \begin{align*}
    &V_i(\theta)= \E\left[X^\theta_{i,T} - L^\theta_{i,T}\,\vert\,\sF_0\right]\\
    &= \frac{n}{2} \Shat_0 - \lambda\E\left[\theta_T^\uparrow+\theta_T^\downarrow\,\vert\,\sF_0\right] \\
    &\ \ \ \ \ + \E\left[\int_0^T \kappa(t)\left(\left(c_0(t)+\gamma(t)c_2\left(a_i-\frac{n}{2}\right)\right)\theta_t -c_1\theta_t^2 -\frac12\left(\theta_t-\frac{n}{2}-\gamma(t)\left(a_i-\frac{n}{2}\right)\right)^2\right)dt\,\vert\,\sF_0\right]\\
    &= \frac{n}{2} \Shat_0 - \lambda\E\left[\theta_T^\uparrow+\theta_T^\downarrow\,\vert\,\sF_0\right] -\frac12 \E\left[\int_0^T\kappa(t)\left(\frac{n}{2}+\gamma(t)\left(a_i-\frac{n}{2}\right)\right)^2dt\,\vert\,\sF_0\right] \\
    &\ \ \ \ \ + \E\left[\int_0^T \kappa(t)\left(\left(\left(c_1+\frac{1}{2}\right)n+\gamma(t)\frac{(1+2c_1)A_i}{1+c_1}\right)\theta_t - \left(c_1+\frac12\right)\theta_t^2\right)dt\,\vert\,\sF_0\right].
  \end{align*}
  The above calculation reveals that when optimizing over $\theta$ in~\eqref{eqn:V}, the second-order condition requires that $c_1>-\frac12$.

  For notational convenience, we define
$$
  \alpha(t):= \left(c_1+\frac{1}{2}\right)n+\gamma(t)\frac{(1+2c_1)A_i}{1+c_1}
  \ \ \ \text{ and } \ \ \ 
  \beta:= c_1+\frac12,
$$
so that
\begin{align*}
  V_i(\theta) 
    &= \frac{n}{2} \Shat_0 - \lambda\E\left[\theta_T^\uparrow+\theta_T^\downarrow\,\vert\,\sF_0\right] \\
    &\ \ \ \ \ -\frac12 \int_0^T\kappa(t)\left(\frac{n}{2}+\gamma(t)\left(a_i-\frac{n}{2}\right)\right)^2 dt
  + \E\left[\int_0^T \kappa(t)\left(\alpha(t)\theta_t - \beta\theta_t^2\right)dt\,\vert\,\sF_0\right].
\end{align*}
Therefore, by algebraic manipulation and applying \eqref{eqn:FOC-calc-2},
\begin{align*}
  V_i&(\theta)-V_i(\theta_i)\\
  &= \E\left[\lambda\left(\theta_{i,T}^\uparrow+\theta_{i,T}^\downarrow-\theta_{T}^\uparrow-\theta_{T}^\downarrow\right)
    +\int_0^T \kappa(u)\left(\alpha(u)(\theta_u-\theta_{i,u})-\beta(\theta_u-\theta_{i,u})(\theta_u+\theta_{i,u})\right)du \right]\\
  &= \E\left[\lambda\left(\theta_{i,T}^\uparrow+\theta_{i,T}^\downarrow-\theta_{T}^\uparrow-\theta_{T}^\downarrow\right)
    +\int_0^T \kappa(u)\left((\alpha(u)-2\beta\theta_{i,u})(\theta_u-\theta_{i,u})-\beta(\theta_u-\theta_{i,u})^2\right)du \right]\\
  &\leq \E\left[\lambda\left(\theta_{i,T}^\uparrow+\theta_{i,T}^\downarrow-\theta_{T}^\uparrow-\theta_{T}^\downarrow\right)
    +\int_0^T \kappa(u)(\alpha(u)-2\beta\theta_{i,u})(\theta_u-\theta_{i,u})du \right]\\
  &= \E\left[\lambda\left(\theta_{i,T}^\uparrow+\theta_{i,T}^\downarrow-\theta_{T}^\uparrow-\theta_{T}^\downarrow\right)
    +\frac{A_i(1+2c_1)}{1+c_1}\int_0^T \kappa(u)\left(\gamma(u)-\tilde\gamma(u)\right)(\theta_u-\theta_{i,u})du \,\vert\,\sF_0\right]\\
  &\leq \E\left[\lambda\left(\theta_{i,T}^\uparrow+\theta_{i,T}^\downarrow\right) - \lambda\left(\theta_{i,T}^\uparrow+\theta_{i,T}^\downarrow\right)\,\vert\,\sF_0\right] 
  \ \ \ \ \text{by \eqref{eqn:FOC-calc-2}}\\
  &\leq 0.
\end{align*}
Thus, $\theta_i$ is optimal for agent $i$. 

\noindent{\it Step 4: Markets clear.} For $t< \tau$, the stock market clears since $\theta_{1,t}+\theta_{2,t} = n + A_1 + A_2 = n$.  For $t\geq\tau$ and $\chi\leq\lambda$, we have $\theta_{1,t}+\theta_{2,t} = n$.  Similarly, for $t\geq\tau$ and $\chi>\lambda$, the stock market clears since
\begin{align*}
  \theta_{1,t}&+\theta_{2,t}\\
  &= n+\frac{1}{\int_\tau^T\kappa(u)du}\left\{\int_\tau^T\kappa(u)\gamma(u)\frac{A_1+A_2}{1+c_1}du - \frac{\lambda\left(\sgn(A_1)+\sgn(A_2)\right)}{1+2c_1}\right\}\\
  &= n+0 + 0 = n.
\end{align*}

\noindent{\it Step 5: Equilibrium prices are consistent.} By the martingale representation of the dividend $D$, we notice that $\Shat_T=D$, as desired.  Finally, we verify that for $t\in\{u\in[0,T]:\, d\theta_{1,u}^\uparrow + d\theta_{1,u}^\downarrow >0\}$, we have that
\begin{equation}\label{eqn:S-equals}
  S^{\theta_1}_{1,t} = S^{\theta_2}_{2,t} = \Shat_t.
\end{equation}
Since $t\in\{u\in[0,T]:\, d\theta_{1,u}^\uparrow + d\theta_{1,u}^\downarrow >0\} \subseteq [0,\tau]$, it is sufficient to verify that \eqref{eqn:S-equals} holds for $t\in[0,\tau]$.

For $i\in\{1,2\}$ and $t\in(0,\tau)$, the drift of $S_{i,t}^{\theta_i}$ is given by
\begin{align*}
  \text{drift}\left(S_{i,t}^{\theta_i}\right)
  &= \kappa(t)\left(\frac{c_1n}{2} - \frac{1}{2}\gamma(t)\left(a_\Sigma-n\right)\right),
\end{align*}
which agrees with the drift of $\Shat$ in \eqref{def:S-hat}.  Since the martingale components and initial values of $\Shat$, $S_1^\theta$, and $S_2^\theta$ do not depend on the choice of $i=1,2$, or $\theta\in\sA_i$, we conclude that \eqref{eqn:S-equals} holds for all $t\in[0,\tau]$, as desired.

\end{proof}

Next, we prove Proposition~\ref{prop:optimal-lambda}, which establishes the existence of an optimal transaction cost, $\hatlambda>0$.
\begin{proof}[Proof of Proposition~\ref{prop:optimal-lambda}]
Theorem~\ref{thm:main} establishes the existence of an equilibrium with optimal trading strategies $\theta_i^{(\lambda)}\in\sA_i$ given in~\eqref{def:theta} for every $\lambda>0$.

As in the proof of Theorem~\ref{thm:main}, we introduce the function $\tilde\gamma$ for $t\in[0,T]$ and $\lambda>0$ by
\begin{equation*}
  \tilde\gamma(t,\lambda):=
  \begin{cases}
    \gamma(t,\lambda), & t<\tau(\lambda),\\
    \frac{\int_{\tau(\lambda)}^T \kappa(u)\gamma(u)du-\frac{\lambda(1+c_1)}{|A_1(1+2c_1)|}}{\int_{\tau(\lambda)}^T\kappa(u)du}, & t\geq\tau(\lambda) \text{ and } \chi>\lambda,\\
    0, & t\geq\tau(\lambda) \text{ and } \chi\leq\lambda.
  \end{cases}
\end{equation*}
Here, $\chi$ is defined in~\eqref{def:chi}, and we make $\lambda$'s dependence on $\tilde\gamma$ explicit.  For any $\lambda>0$, we have that
$$
  \theta_{1,T}^{(\lambda)}-\frac{n}{2} = \theta_{1,\tau(\lambda)}^{(\lambda)}-\frac{n}{2}
  = \frac{A_1}{1+c_1} \tilde\gamma(\tau(\lambda),\lambda).
$$
We first seek to show that $\lambda\mapsto\tilde\gamma(\tau(\lambda),\lambda)$ is continuous, from which we will conclude that $\lambda\mapsto\text{Profit}(\lambda)$ is continuous.  On $\{A_1=0\}$, $\chi=0$ so that $\chi<\lambda$ for all $\lambda>0$, and thus, $\lambda\mapsto\tilde\gamma(\tau(\lambda),\lambda)$ is continuous in this case.

On $\{A_1\neq 0\}$, we consider $0<\lambda_1<\lambda_2$.  Then, $\tau(\lambda_2)\leq \tau(\lambda_1)< T$.  For notational simplicity, we define $C:= \frac{1+c_1}{|A_1|(1+2c_1)}$ and $\tau_1:=\tau(\lambda_1)$, $\tau_2:=\tau(\lambda_2)$.  We have that
\begin{align*}
  \lambda_1 C
  &= \int_{\tau_1}^T \kappa(u)\Big(\gamma(u)-\tilde\gamma(\tau_1,\lambda_1)\Big)du \\
  &= \int_{\tau_1}^{\tau_2} \kappa(u)\Big(\gamma(u)-\tilde\gamma(\tau_1,\lambda_1)\Big)du \\
  &\ \ \ \ \ \ \ \ \ \ \  + \underbrace{\int_{\tau_2}^T \kappa(u)\Big(\gamma(u)-\tilde\gamma(\tau_2,\lambda_2)\Big)du}_{=\lambda_2 C} 
     + \int_{\tau_2}^{T} \kappa(u)\Big(\tilde\gamma(\tau_2,\lambda_2)-\tilde\gamma(\tau_1,\lambda_1)\Big)du \\
  &= \int_{\tau_2}^{\tau_1} \kappa(u)\Big(\tilde\gamma(\tau_1,\lambda_1)-\gamma(u)\Big)du + \lambda_2 C
     + \Big(\tilde\gamma(\tau_2,\lambda_2)-\tilde\gamma(\tau_1,\lambda_1)\Big) \int_{\tau_2}^T\kappa(u)du.
\end{align*}
Using that $\gamma$ is nondecreasing, we see that $\tilde\gamma(\tau_2,\lambda_2)\leq\gamma(u)$ for all $u\in(\tau_2,\tau_1)$.  By rearranging terms, we arrive at
\begin{align*}
  0 &\leq \Big(\tilde\gamma(\tau_1,\lambda_1)-\tilde\gamma(\tau_2,\lambda_2)\Big) \int_{\tau_2}^T\kappa(u)du \\
  &= C(\lambda_2-\lambda_1) + \int_{\tau_2}^{\tau_1} \kappa(u)\Big(\tilde\gamma(\tau_1,\lambda_1)-\gamma(u)\Big)du \\
  &\leq C(\lambda_2-\lambda_1) +  \Big(\tilde\gamma(\tau_1,\lambda_1)-\tilde\gamma(\tau_2,\lambda_2)\Big)\int_{\tau_2}^{\tau_1}\kappa(u)du.
\end{align*}
Therefore, on $\{A_1\neq 0\}$, we have that
\begin{equation}\label{eqn:estimate}
  0 \leq \tilde\gamma(\tau_1,\lambda_1)-\tilde\gamma(\tau_2,\lambda_2)
  \leq \frac{C}{\int_{\tau_1\vee\tau_2}^T \kappa(u)du}(\lambda_2-\lambda_1),
\end{equation}
from which we conclude that $\lambda\mapsto\tilde\gamma(\tau(\lambda),\lambda)$ is continuous for $\lambda>0$.  Thus, $\lambda\mapsto \theta^{(\lambda)}_{1,T} - \frac{n}{2}$ and $\lambda\mapsto \text{Profit}(\lambda)$ are continuous for $\lambda>0$.

Next, we show that $\lambda\mapsto\E[\text{Profit}(\lambda)]$ is continuous.  For $0<\lambda_1<\lambda_2$, we fix constants $\underline\lambda, \overline\lambda$ so that $0<\underline\lambda\leq \lambda_1<\lambda_2\leq \overline\lambda$.  We let $\tau_1:=\tau(\lambda_1)$ and $\tau_2:=\tau(\lambda_2)$, as before.  Using~\eqref{eqn:estimate} and that $\gamma$ is nondecreasing, we see that
\begin{align*}
  \Big\vert \text{Profit}(\lambda_1)&- \text{Profit}(\lambda_2)\Big\vert
  = 2\Bigg\vert 
    \lambda_1\left\lvert\theta_{1,\tau_1}^{(\lambda_1)}-\frac{n}{2}\right\rvert 
    - \lambda_2\left\lvert\theta_{1,\tau_2}^{(\lambda_2)}-\frac{n}{2}\right\rvert
    \Bigg\vert \\
  &\leq 2\lambda_2\left\lvert \theta_{1,\tau_1}^{(\lambda_1)}-\theta_{1,\tau_2}^{(\lambda_2)}\right\rvert 
    + 2\left(\lambda_2-\lambda_1\right)\left\lvert \theta_{1,\tau_1}^{(\lambda_1)}-\frac{n}{2}
  \right\rvert \\
  &\leq 2\left(\frac{\overline\lambda}{(1+2c_1)\int_{\tau(\underline\lambda)}^T\kappa(u)du}+\frac{|A_1|}{1+c_1}\gamma(T)\right)\left(\lambda_2-\lambda_1\right).
\end{align*}
Taking expectations shows that for all $0<\underline\lambda\leq \lambda_1<\lambda_2\leq\overline\lambda<\infty$, we have
$$
  \E\Big[\left\lvert \text{Profit}(\lambda_1)- \text{Profit}(\lambda_2)\right\rvert\Big]
  \leq 2\left(\frac{\overline\lambda}{(1+2c_1)\int_{\tau(\underline\lambda)}^T\kappa(u)du}+\frac{\E[|A_1|]}{1+c_1}\gamma(T)\right)\left(\lambda_2-\lambda_1\right),
$$
which proves that $\lambda\mapsto\E[\text{Profit}(\lambda)]$ is continuous for $\lambda>0$.

 Finally, we show that $\lambda\mapsto\E[\text{Profit}(\lambda)]$ achieves a maximum for $\lambda>0$.  The assumption that $\E[A_1^2]>0$ and that there exists $t\in[0,T)$ for which $\gamma(t)>0$ ensures that there exists $\lambda>0$ such that $\E[\text{Profit}(\lambda)]>0$.  Since $\tilde\gamma$ is bounded by one, we have that
 $$
   0\leq \lim_{\lambda \downarrow 0} \E\left[\text{Profit}(\lambda)\right]
   \leq \lim_{\lambda \downarrow 0} 2\lambda\ \E\left[\frac{|A_1|}{1+c_1}\right]
   = 0.
 $$
 We define the random variable $\lambdamax := 2|A_1| \int_0^T \kappa(u)du$.  For $\lambda\geq\lambdamax$, we have that $\chi<\lambda$, where $\chi$ is defined in~\eqref{def:chi}. In this case, trade does not occur, and $\left\lvert\theta^{(\lambda)}_{1,\tau(\lambda)} - \frac{n}{2}\right\rvert = 0$.  
For any $\lambda>0$, we bound $\text{Profit}(\lambda)$ by
\begin{align*}
  \text{Profit}(\lambda)
  &\leq \frac{2\lambda |A_1|}{1+c_1} \gamma(T)\bI_{\{\lambda\leq\lambdamax\}} \\
  &\leq \frac{2\lambdamax |A_1|}{1+c_1}\\
  &\leq A_1^2 \cdot\frac{4\int_0^T\kappa(u)du}{1+c_1}.
\end{align*}
Since $\E[A_1^2]<\infty$, we apply the dominated convergence theorem to obtain that
$$
  \lim_{\lambda\rightarrow\infty} \E[\text{Profit}(\lambda)] = 0.
$$
Thus, $\lambda\mapsto\E[\text{Profit}(\lambda)]$ achieves a maximum for $\lambda>0$.

%
 
\end{proof}

\bibliographystyle{plain}
\bibliography{finance_bib}

\end{document}